\documentclass[12pt]{article}       

\usepackage{amsthm}
\usepackage{amsmath}
\usepackage{graphicx}
\usepackage{amssymb}

\newtheorem{thm}{Theorem}[section]
\newtheorem{lma}[thm]{Lemma}
\newtheorem{cor}[thm]{Corollary}
\newtheorem{conj}{Conjecture}
\newtheorem{prob}{Problem}

\begin{document}

\title{The complexity of flood-filling games on graphs}
\date{\today}
\author{Kitty Meeks and Alexander Scott\\
\small{Mathematical Institute, University of Oxford, 24-29 St Giles', Oxford OX1 3LB, UK} \\
\texttt{\small{\{meeks,scott\}@maths.ox.ac.uk}}}
\maketitle

\begin{abstract}
We consider the complexity of problems related to the combinatorial game Free-Flood-It, in which players aim to make a coloured graph monochromatic with the minimum possible number of flooding operations.  Although computing the minimum number of moves required to flood an arbitrary graph is known to be NP-hard, we demonstrate a polynomial time algorithm to compute the minimum number of moves required to link each pair of vertices.  We apply this result to compute in polynomial time the minimum number of moves required to flood a path, and an additive approximation to this quantity for an arbitrary $k \times n$ board, coloured with a bounded number of colours, for any fixed $k$.  On the other hand, we show that, for $k \geq 3$, determining the minimum number of moves required to flood a $k \times n$ board coloured with at least four colours remains NP-hard.
\end{abstract}

\section{Introduction}

In this paper we consider the complexity of a number of problems related to the one-player combinatorial game Flood-It, first studied by Arthur, Clifford, Jalsenius, Montanaro and Sach in \cite{arthurFUN}.  The original game is played on a board consisting of an $n \times n$ grid of coloured squares, where each square is given a colour from some fixed colour-set.  The player makes a move by changing the colour of the monochromatic path-connected area containing the top left square, and the goal is to make the entire board monochromatic with the minimum possible number of such moves.  We also consider the ``free'' variant of Flood-It in which at each move, as well as choosing a colour, the player can choose freely which area's colour to change.

The game can more generally be played on any graph $G$ equipped with a colouring $\omega$.  Here, in the free version, a move consists of choosing a vertex $v$ and a colour $d$, and giving all vertices in the same monochromatic component as $v$ colour $d$.  Alternatively, we may always play moves at some fixed vertex, as in the original version of the game.  Again, in either case, the aim is to make the entire graph monochromatic using as few moves as possible.

For any board or, more generally, coloured graph, we define the following problems.
\begin{itemize}
\item FIXED-FLOOD-IT is the problem of determining the minimum number of moves required to flood any given coloured graph, if we always play at a specified vertex.  The number of colours may be unbounded.
\item FREE-FLOOD-IT is the same problem when we are allowed to make moves anywhere in the graph.
\item $c$-FIXED-FLOOD-IT and $c$-FREE-FLOOD-IT respectively are the variants of FIXED-FLOOD-IT and FREE-FLOOD-IT in which only colours from some fixed set of size $c$ are used.
\end{itemize}

In \cite{arthurFUN}, Arthur, Clifford, Jalsenius, Montanaro and Sach show that, for any $c \geq 3$, $c$-FIXED-FLOOD-IT and $c$-FREE-FLOOD-IT are both NP-hard on a standard $n \times n$ board.  They further show that, unless $\mathbf{P}=\mathbf{NP}$, there can be no constant-factor (independent of the number of colours $c$) polynomial time approximation algorithm.  

We prove a number of results about the game played on both general graphs and paths, and give a polynomial-time algorithm to compute the minimum number of moves required to connect each pair of vertices in a general graph.  Using some of these results, we then consider the game played on a rectangular $k \times n$ board for various fixed values of $k$.

In particular, we prove the following results.
\begin{itemize}
\item 2-FREE-FLOOD-IT is solvable in polynomial time, answering an open question from an earlier version of \cite{arthurFUN} (posted January 2010).

\item In an arbitrary graph $G = (V,E)$ coloured with any colour-set $C$, the number of moves required to connect the vertices $u$ and $v$ can be computed, for every pair $(u,v) \in V^2$, in time $O(|V|^3|E||C|^2)$.

\item FREE-FLOOD-IT, restricted to $1 \times n$ boards, can be solved in polynomial time.

\item We can compute in polynomial time an additive approximation to $c$-FREE-FLOOD-IT, restricted to $k \times n$ boards, for any fixed integers $k$ and $c$.

\item 4-FIXED-FLOOD-IT and 4-FREE-FLOOD-IT remain NP-hard when restricted to $3 \times n$ boards.
\end{itemize}

Two recent papers (\cite{clifford} and \cite{lagoutte}) both independently show our first result, that 2-FREE-FLOOD-IT is polynomially solvable on general graphs.  In \cite{lagoutte}, Lagoutte also shows that FIXED-FLOOD-IT is polynomially solvable on cycles, whereas for $c \geq 3$, $c$-FIXED-FLOOD-IT and $c$-FREE-FLOOD-IT are NP-hard when restricted to trees.  The hardness of $c$-FIXED-FLOOD-IT on trees was shown independently by Fleischer and Woeginger in their analysis of variants of the related Honey-Bee Game \cite{fleischer10}.

Clifford, Jalsenius, Montanaro and Sach give in \cite{clifford} an $O(n)$ algorithm to solve FIXED-FLOOD-IT on $2 \times n$ boards.  In a companion paper \cite{2xn} we complete the picture for such boards by considering the complexity of ($c$-)FREE-FLOOD-IT.  In particular, we show that for any fixed $c$, $c$-FREE-FLOOD-IT is fixed parameter tractable with parameter $c$; on the other hand, FREE-FLOOD-IT remains NP-hard when restricted to $2 \times n$ boards.

We begin in Section \ref{notation} with some notation and definitions, then in Section \ref{2-free} we consider 2-FREE-FLOOD-IT.  In Section \ref{general} we derive results for general graphs and apply them to the cases of $1 \times n$ and $k \times n$ boards, before showing the complexity results for $3 \times n$ boards in Section \ref{complexity}.

\section{Notation and Definitions}
\label{notation}

Although the original Flood-It game is played on a square grid, we can more generally consider the same game played on any graph $G=(V,E)$, with an initial colouring using colours from the \emph{colour-set} $C$.  Then each move $m=(v,d)$ consists of choosing some vertex $v \in V$ and a colour $d \in C$, and assigning colour $d$ to all vertices in the same monochromatic component as $v$.  The goal is to make every vertex in $G$ the same colour, using as few moves as possible.  We may assume, without loss of generality, that the initial colouring is proper: if not, we simply contract each monochromatic component to a single vertex.

Given any connected graph $G$, equipped with a proper colouring $\omega$, we define $m(G,\omega,d)$ to be the minimum number of moves required to give all its vertices colour $d$, and $m(G,\omega)$ to be $\min_{d \in C}m(G,\omega,d)$.  For any subgraph $H$ of $G$, we write $\omega |_H$ for the colouring $\omega$ restricted to $H$.  Given any sequence of moves $S$ on a graph $G$ with initial colouring $\omega$, we denote by $S(\omega,G)$ (or simply $S(\omega)$ if $G$ is clear from the context) the new colouring obtained by playing $S$ in $G$.  

Let $A$ be any subset of $V$.  We then say a move $m=(v,c)$ is \emph{played in} $A$ if $v \in A$, and that $A$ is \emph{linked} if it is contained in a single monochromatic component.  The \emph{(edge) boundary} of $A$ is defined to be the set of edges $b = \{uv \in E: u \in A, v \notin A\}$, and we say that $A_1, A_2 \subseteq V$ are \emph{adjacent} if their edge boundaries have nonempty intersection.  We call any connected induced subgraph of $G$ an \emph{area}.

When we consider the game played on a rectangular board $B$, we are effectively playing the game in a graph $G_B$ with an initial (proper) colouring $\omega_B$.  This graph is obtained from the planar dual of $B$ (in which there is one vertex corresponding to each square of $B$, and vertices are adjacent if they correspond to squares which are either horizontally or vertically adjacent in $B$) by giving each vertex the colour of the corresponding square in $B$, and contracting every monochromatic component to a single vertex.  We define a \emph{region} of the board $B$ to be a collection of squares corresponding to a single vertex in $G_B$, and thus regions are fixed by the initial colouring.  We shall sometimes use $B$ as a shorthand for $G_B, \omega_B$ (writing, for example, $m(B)$ rather than $m(G_B, \omega_B)$).

\section{2-FREE-FLOOD-IT is solvable in polynomial time}
\label{2-free}

In this section we consider the free version of two-colour Flood-It, played on an arbitrary connected graph $G=(V,E)$.  When making a move $m=(v,d)$ in such a game, our only choice is the vertex at which we play, as there is only one possible way to change its colour.  Making a move in the game then corresponds to picking a vertex $v \in V$ and contracting all edges incident with it, and the aim of the game is to reduce the graph to a single vertex with as few moves as possible.  We can then regard any strategy as a sequence of vertices around which we perform contractions.  Of course, we may contract at a vertex $w$ which was created by an earlier contraction, but in this case we can always choose a vertex $u$ from the original graph as a representative for $w$, and regard the contraction as being performed about $u$.

\begin{lma}
There exists an optimal strategy in which we contract at the same vertex in every move.
\label{single-vertex}
\end{lma}

\begin{proof}
Suppose that, for some $v_1, \ldots, v_k \in V$, $S = v_1^{r_1} \ldots v_k^{r_k}$ is an optimal sequence, with $k$ as small as possible (where we perform $r_i$ consecutive contractions about the vertex $v_i$).  We will show that, if $k \geq 2$, there exists a sequence of moves, of no greater length, which contracts the graph to a single vertex and in which contractions are performed about only $k-1$ distinct vertices, contradicting the minimality of $k$ and thus proving the result.

Let us denote by $G'$ the graph obtained by performing the sequence of contractions $v_1^{r_1} \ldots v_{k-2}^{r_{k-2}}$, so the remaining contractions about $v_{k-1}$ and $v_k$ reduce $G'$ to a single vertex.  We claim that there exists a single vertex $w$ such that all vertices in $G'$ are at distance at most $r_{k-1}+r_k$ from $w$, and hence we can perform $r_{k-1} + r_k$ contractions about $w$ to reduce $G'$ to a single vertex, giving our contradiction to the minimality of $k$.

Consider a shortest path $P$ from $v_{k-1}$ to $v_k$ in $G'$.  Without loss of generality we may assume $d(v_{k-1},v_k) > r_{k-1}$, otherwise $v_k$ is absorbed by the contractions performed around $v_{k-1}$ and, in order to minimise the number of distinct vertices, we would have chosen $v_{k-1}$ as a representative for the vertex about which we perform the remaining contractions.  Observe also that the length of $P$ is at most $r_{k-1} + r_k$, or the two final sets of contractions would not reduce $P$ to a single vertex.  

Let $\alpha = r_{k-1} + r_k - d(v_{k-1},v_k) \geq 0$.  We can then consider the last $r_{k-1} + r_k$ moves of $S$ in three stages.  
\begin{enumerate}
\item The first $r_{k-1}$ moves contract all vertices at distance at most $r_{k-1}$ from $v_{k-1}$ in $G'$ to a single vertex, $u_1$, in the new graph $G_1$.
\item The next $r_k - \alpha$ moves contract all vertices at distance at most $r_k - \alpha$ from $v_k$ in $G_1$ to a single vertex $u_2$ in the new graph $G_2$.  Note that $u_1$ is absorbed only at the final step.
\item The remaining $\alpha$ moves absorb only vertices within distance $\alpha$ of $u_2$ in $G_2$.  Thus we absorb any vertices at distance at most $r_k$ from $v_k$ in $G'$, and additionally any other vertices at distance at most $\alpha$ from $u_1$ in $G_1$, that is vertices at distance at most $r_{k-1} + \alpha$ from $v_{k-1}$ in $G'$.
\end{enumerate}
Hence, as these $r_{k-1} + r_k$ moves reduce $G'$ to a single vertex, we know that for every vertex $x \in G'$, either $d(x,v_k) \leq r_k$, or $d(x,v_{k-1}) \leq r_{k-1} + \alpha$.

Now set $w$ to be the vertex on $P$ at distance $r_{k-1}$ from $v_k$.  It remains to check that if $d(x,v_k) \leq r_k$ or $d(x,v_{k-1}) \leq r_{k-1} + \alpha$ then we have $d(x,w) \leq r_{k-1} + r_k$. 

First suppose $d(x,v_k) \leq r_k$.  Then
$$d(x,w) \leq d(x,v_k) + d(v_k,w) \leq r_k + r_{k-1},$$
as required.  Now suppose that $d(x,v_{k-1}) \leq r_{k-1} + \alpha$.  But then we have
\begin{align*}
d(x,w) & \leq d(x,v_{k-1}) + d(v_{k-1},w)  \\
	   & \leq r_{k-1} + \alpha + d(v_{k-1},v_k) - r_{k-1} \\
	   & = r_{k-1} + r_k,
\end{align*}
as required.
\end{proof}

\begin{thm}
2-FREE-FLOOD-IT is solvable in polynomial time on arbitrary graphs.
\end{thm}

\begin{proof}
By Lemma \ref{single-vertex}, it is enough to consider strategies in which we contract about the same vertex in every move.  It is clear that the number of moves required, if we always contract around the vertex $v$, is equal to $\max_{u \in V(G)} d(u,v)$, and that the minimum number of moves required to flood the entire graph is obtained by taking the minimum over all possible vertices $v$.  But this is exactly equal to the radius of the graph, which can easily be computed in polynomial time.
\end{proof}

\section{General results for Free-Flood-It}
\label{general}

The main result of this section is a polynomial-time algorithm to determine the minimum number of moves required to link $u$ and $v$, for every pair of vertices $(u,v)$ in an arbitrary connected graph.  We begin by proving two auxiliary results about the special case in which the game is played on a path, and then apply these results to sequences of moves linking pairs of vertices in arbitrary connected graphs.  

We start with a monotonicity result for paths.

\begin{lma}
Let $P$ be a path, with colouring $\omega$ from colour-set $C$, and let $P'$ be a second coloured path with colouring $\omega'$, obtained from $P$ by deleting one vertex and joining its neighbours.  Then, for any $d \in C$, $m(P',\omega',d) \leq m(P,\omega,d)$.  We also have $m(P',\omega') \leq m(P,\omega)$.
\label{monotonicity}
\end{lma}

\begin{proof}
Fix $d \in C$, and note we may assume that $\omega$ is a proper colouring of $P$ (contracting monochromatic components if necessary, and observing that the result is trivially true if $v$ has a neighbour of the same colour).  We proceed by induction on $m(P,\omega,d)$.  The result is trivially true for $m(P,\omega,d) = 0$, so assume $m(P,\omega,d) \geq 1$ and that the result holds for any path $Q$ with colouring $\omega_Q$ such that $m(Q,\omega_Q,d) < m(P,\omega,d)$. Let $S$ be an optimal sequence to flood $P$ in colour $d$, and let $\alpha$ be the first move of $S$.  Suppose that $V(P') = V(P) \setminus \{v\}$, and that $E(P')=(E(P)\setminus \{uv: u \in V(P)\}) \cup \{uw: u \neq w \in \Gamma(v)\}$.

First suppose that $\alpha$ is not played at the vertex $v$.  Then we can play $\alpha$ on $P'$, and the path $P'$ with colouring $\alpha(\omega',P')$ is identical to that obtained from $P$ with colouring $\alpha(\omega,P)$ by deleting the vertex $v$ and joining its neighbours.  Moreover, $m(P,\alpha(\omega,P),d) < m(P,\omega,d)$, and so by the inductive hypothesis we have $m(P',\alpha(\omega',P'),d) \leq m(P,\alpha(\omega,P),d)$.  Thus
$$m(P',\omega',d) \leq 1 + m(P',\alpha(\omega',P'),d) \leq 1 + m(P,\alpha(\omega,P),d) = m(P,\omega,d),$$
as required.

Now suppose $\alpha$ is played at $v$.  Then the path obtained from $P$ with colouring $\alpha(\omega,P)$ by deleting the vertex $v$ and joining its neighbours gives the path $P'$ with colouring $\omega'$, since $\omega$ is a proper colouring and so changing the colour of $v$ cannot change the colour of any other vertex. Hence, as $m(P,\alpha(\omega),d) < m(P,\omega,d)$ we have, by the inductive hypothesis,
$$m(P',\omega',d) \leq m(P, \alpha(\omega,P),d) < m(P,\omega,d).$$

Thus in all cases we have $m(P',\omega',d) \leq m(P,\omega,d)$, and as this holds for any colour $d \in C$ it follows immediately that $m(P',\omega') \leq m(P,\omega)$.
\end{proof}

We also need a simple fact about additivity.

\begin{lma}
Let $P_1$ and $P_2$ be paths, with colourings $\omega_1$ and $\omega_2$ from colour-set $C$, let $P=P_1P_2$ be the path obtained by concatenating $P_1$ and $P_2$, and let $\omega$ be the colouring of $P$ which agrees with $\omega_i$ on $P_i$.  Then, for any $d \in C$, $m(P,\omega,d) \leq m(P_1,\omega_1,d) + m(P_2,\omega_2,d)$.
\label{cat-paths}
\end{lma}

\begin{proof}
For $i \in \{1,2\}$, let $S_i$ be an optimal sequence to make $P_i$ monochromatic with colour $d$.  Suppose we begin by playing the sequence $S_1$ on $P$.  This makes $P_1$ monochromatic with colour $d$, and may also absorb some vertices from $P_2$.  But by Lemma \ref{monotonicity}, we can make $P_2'$, the remainder of $P_2$, monochromatic in colour $d$ with a sequence $T_2$ of at most $|S_2|$ moves.  In the course of $T_2$, some vertex on $P_2$ may absorb $P_1$, but as this vertex ends up with colour $d$, the sequence $S_1T_2$ must ultimately give $P_1$ colour $d$.  Hence $m(P,\omega,d) \leq |S_1| + |S_2| = m(P_1,\omega_1,d) + m(P_2,\omega_2,d)$.
\end{proof}

Before moving on to the general case, we need a few further definitions.  Suppose $G = (V,E)$ is a connected graph, with colouring $\omega$ from colour set $C$, and let $u,v \in V$.  Then, for any $d \in C$, we define $m_{G,\omega}(u,v,d)$ to be the minimum number of moves required to link $u$ and $v$ in $G$ with a monochromatic path of colour $d$.  We then set $m_{G,\omega}(u,v) = \min_{d \in C} m_{G,\omega}(u,v,d)$.  When it is clear from the context which graph $G$ and colouring $\omega$ are being considered, we may simply write $m(u,v,d)$ or $m(u,v)$.  

Given two vertices $u,v \in V$, we define $\mathcal{P}_G(u,v)$ to be the set of all $u$-$v$ paths in $G$.  If $S$ is a sequence of moves linking two vertices $u$ and $v$, we say that $P \in \mathcal{P}_G(u,v)$ is \emph{critical} with respect to $S$ if, for all $x,y$ lying on $P$, $S$ does not link $x$ and $y$ in $G$ before they are linked along $P$.

\begin{lma}
Let $G$ be a connected graph, with colouring $\omega$, and suppose that $S$ is a sequence of moves linking the vertices $u_1$ and $w_2$.  Then there exists a critical $u_1$-$w_2$ path with respect to $S$.
\label{general-nice-path}
\end{lma}

\begin{proof}
We proceed by induction on $|S|$.  The base case for $|S| = 1$ is trivially true, so we assume $|S| > 1$.

Let $m$ be the first move in $S$ that links $u_1$ and $w_2$.  Denote by $S_1$ the initial segment of moves in $S$ occurring before $m$, and by $S_2$ those occurring after $m$, so $S=S_1mS_2$.  Let $U$ be the maximal monochromatic area containing $u_1$ immediately before $m$, and $W$ the maximal monochromatic area containing $w_2$ at this point.  There are two cases: either $U$ and $W$ are adjacent, or there is some third monochromatic area $V$, adjacent to both, such that $m$ changes the colour of $V$ to be the same as that of both $U$ and $W$.

First suppose we are in the second situation, so $m$ changes the colour of a third area, $V$, to link $U$ and $W$.  Let $S_U$, $S_V$ and $S_W$ be the subsequences of $S_1$ consisting of moves played in the areas $U$, $V$ and $W$ respectively.  As $U$, $V$ and $W$ are maximal monochromatic areas, no move from any of the subsequences has any effect on vertices outside the area in which it is played, and so the subsequences are disjoint.  Pick $u_2 \in U$, $w_1 \in W$ and $v_1,v_2 \in V$ such that $u_2$ is adjacent to $v_1$ and $v_2$ is adjacent to $w_1$ (note that there must exist at least one possible choice for each of these vertices, as $U$, $V$ and $V$, $W$ are adjacent).

Clearly if we play $S_U$ in $G[U]$ (with colouring $\omega|_U$) then this links $u_1$ and $u_2$, and similarly $S_V$ links $v_1$ and $v_2$ in $G[V]$ and $S_W$ links $w_1$ and $w_2$ in $G[W]$.  Moreover, as each of these sequences is strictly shorter than $S$, we can apply the inductive hypothesis to obtain a $u_1$-$u_2$ path $P_U$ in $G[U]$ such that no pair of vertices on $P_U$ is linked in $G[U]$ before it is linked along $P_U$.  In the same way we obtain $v_1$-$v_2$ and $w_1$-$w_2$ paths $P_V$, $P_U$ in $G[V]$, $G[W]$ respectively.

\begin{figure} [h]
\centering
\includegraphics[width=0.6\linewidth]{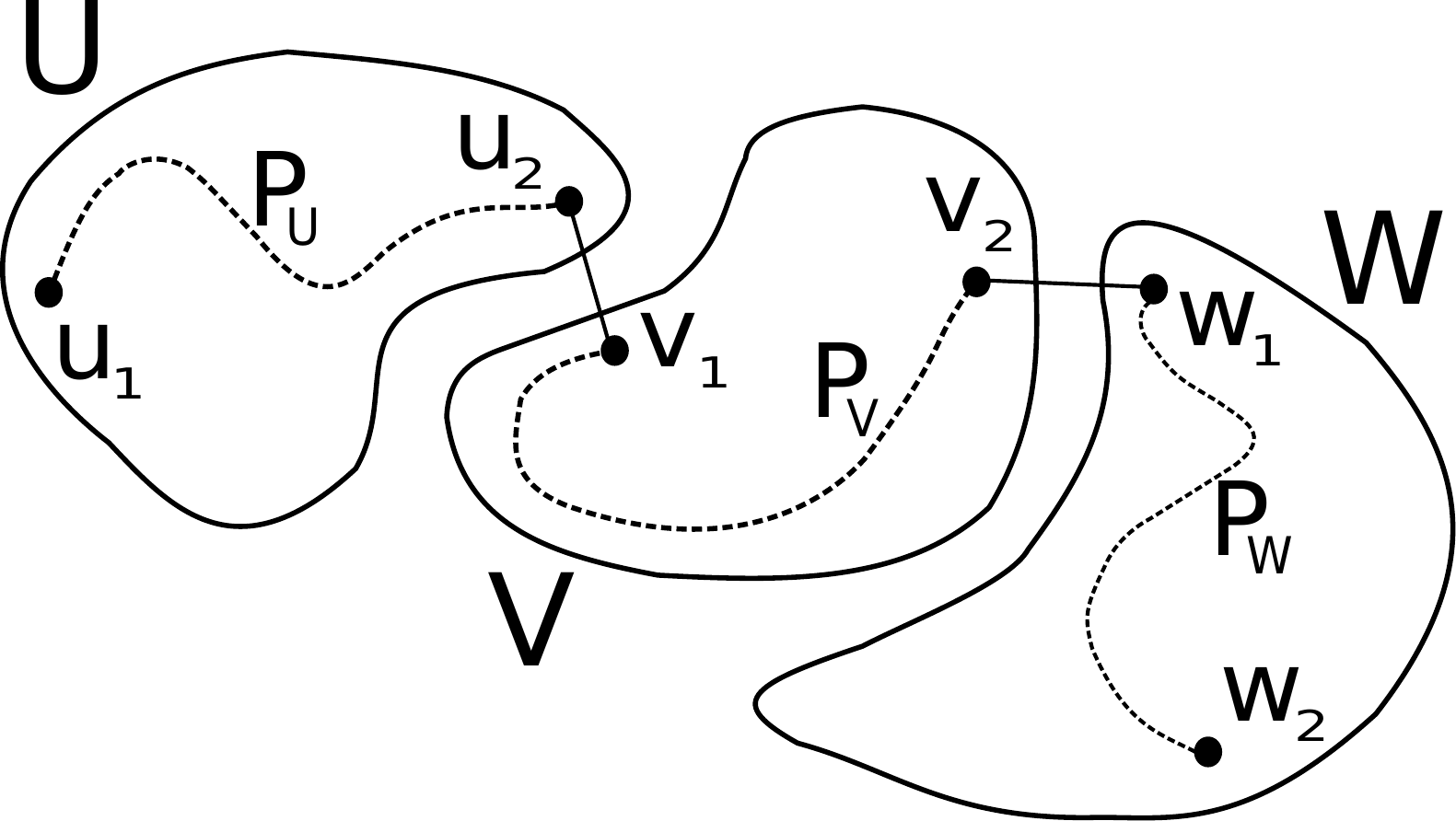}
\caption{The path $P$}
\label{critical-path}
\end{figure}

Now define $P$ to be $u_1P_Uu_2v_1P_Vv_2w_1P_Ww_2$, as illustrated in Figure \ref{critical-path}.  First observe that $P$ is indeed a path: as $U$, $V$ and $W$ are disjoint, no vertex may be repeated.  We claim that $P$ is the path we require.  For, if not, there exist vertices $x$ and $y$ on $P$ and a move $m' \in S$ such that $m'$ links $x$ and $y$ in $G$ before they are linked along $P$.  Clearly $m'$ cannot be from $S_2$, as $P$ is monochromatic before any move in $S_2$ is made, so $x$ and $y$ are already linked along the path by this point.  Nor can we have $m'=m$, since if $x$ and $y$ are not already linked along $P$ before $m$, then $m$ links them along $P$.  So $m' \in S_1$.  But then $x$ and $y$ are linked before $m$, so they must both lie in one of $U$, $V$ and $W$ (as these three areas are not linked immediately before $m$, and moves cannot unlink areas that were previously linked).  Without loss of generality, suppose $x,y \in U$.  But then both $x$ and $y$ lie on $P_U$ so, by definition of $P_U$, they are not linked in $G[U]$ before they are linked along $P_U$, and hence (as $S_U$ has the same effect on $G[U]$ as does the sequence $S_1$ played in $G$) they are not linked in $G$ before they are linked along $P$.  So $P$ is as required.

Now suppose that in fact $U$ and $W$ are adjacent.  We then choose $u_2 \in U$ and $w_1 \in W$ so that $u_2$ and $w_1$ are adjacent (again noting that there must exist such a pair of vertices).  As before, we apply the inductive hypothesis to obtain suitable paths $P_U$ and $P_W$ in $G[U]$ and $G[W]$ respectively, and by exactly the same reasoning the path $P = u_1P_Uu_2w_1P_Ww_2$ is as required.
\end{proof}

We now show that in order to determine $m_{G,\omega}(u,v,d)$, it is enough to consider only $u$-$v$ paths in $G$.

\begin{lma}
Let $G=(V,E)$ be any connected graph with colouring $\omega$ from colour-set $C$, and let $u,v \in V$ and $d \in C$.  Then 
$$m_{G,\omega}(u,v,d) = \min_{P \in \mathcal{P}_G(u,v)} m(P,\omega|_P,d).$$
\label{min-path}
\end{lma}

\begin{proof}
Let $S$ be an optimal sequence to link $u$ and $v$ with colour $d$ in $G$, and let $P$ be a critical $u$-$v$ path with respect to $S$.  Let $S'$ be the subsequence of $S$ consisting of moves played in areas intersecting $P$, and without loss of generality assume every move in $S'$ is played on $P$ (otherwise we may replace it with an equivalent move played on $P$).  Then, as $P$ is critical with respect to $S$, the sequence $S'$ played on the path $P$ (considered as a separate graph) has the same effect on $P$ as does $S$ when played in $G$, and so makes $P$ monochromatic.  Thus,
$$m_{G,\omega}(u,v,d) = |S| \geq |S'| \geq m(P,\omega |_P,d) \geq \min_{P \in \mathcal{P}_G(u,v)} m(P,\omega |_P,d).$$

To show the reverse inequality, we prove by induction on $m(P,\omega |_P,d)$ that, for any $P \in \mathcal{P}_G(u,v)$, $m_{G,\omega}(u,v,d) \leq m(P,\omega |_P,d)$.  The base case, for $m(P,\omega|_P,d) = 0$, is trivially true, so let $P_1 \in \mathcal{P}_{uv}$ and suppose $S_1$ is a nonempty optimal sequence to flood $P_1$ with colour $d$.  Consider the first move, $\alpha$, of $S_1$.

\begin{figure} [h]
\centering
\includegraphics[width=0.5\linewidth]{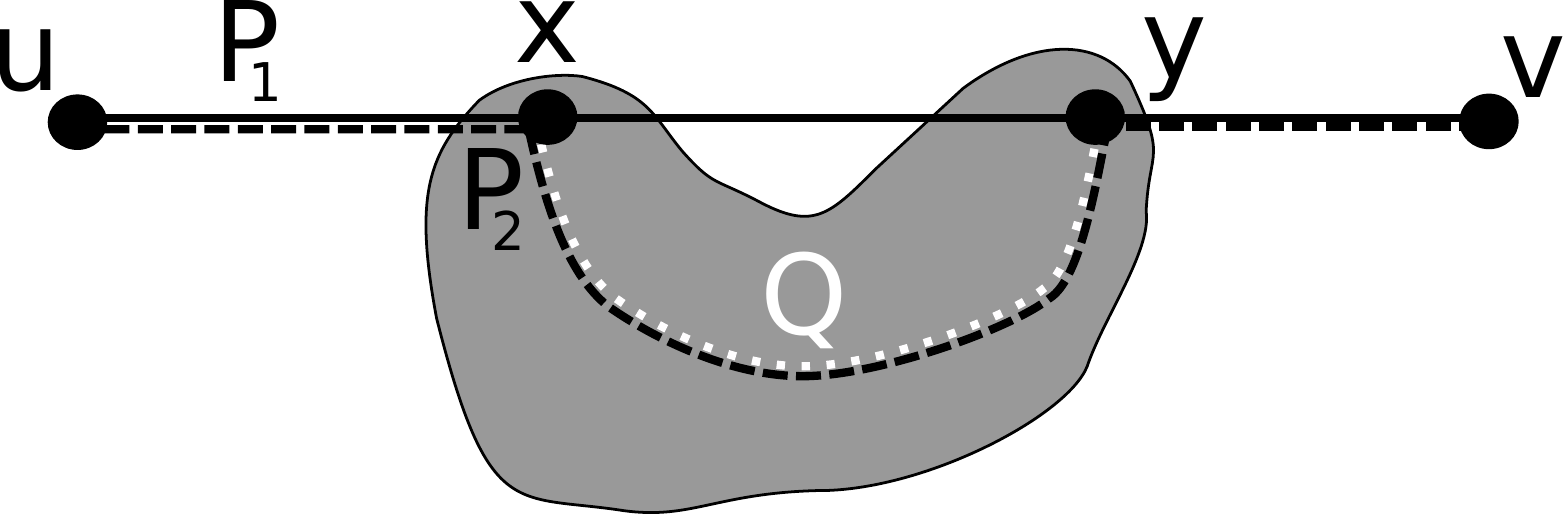}
\caption{The $u$-$v$ paths $P_1$ and $P_2$. The shaded area is a monochromatic component of $G$ with colouring $\alpha(\omega,G)$.}
\label{P_2}
\end{figure}

First suppose that there exist two or more vertices on $P_1$ whose colours are changed by $\alpha$ when the move is played in $G$, but are not linked along the path.  Suppose $x$ is the first such vertex on $P_1$ when the path is traversed from $u$ to $v$, and $y$ the last, and observe that there exists a monochromatic $x$-$y$ path $Q$ in $G$ with colouring $\alpha(\omega,G)$.  Let $P_2$ be the $u$-$v$ path in $G$ obtained by joining the segments of $P_1$ from $u$ to $x$ and from $y$ to $v$ with the path $Q$, as illustrated in Figure \ref{P_2}.  Then the path $P_2$ with colouring $\alpha(\omega,G)|_{P_2}$ can (after contracting monochromatic components) be obtained from $P_1$ with colouring $\alpha(\omega |_{P_1}, P_1)$ by deleting some consecutive vertices and joining the resulting segments so, by Lemma \ref{monotonicity}, $m(P_2,\alpha(\omega,G)|_{P_2},d) \leq m(P_1, \alpha(\omega|_{P_1},P_1),d) = m(P_1,\omega|_{P_1},d) -1$.  Hence, by the inductive hypothesis, we have
$$m_{G,\alpha(\omega,G)}(u,v,d) \leq m(P_2,\alpha(\omega,G)|_{P_2},d),$$
and so
\begin{align*}
m_{G,\omega}(u,v,d) & \leq 1 + m_{G,\alpha(\omega,G)}(u,v,d) \\
				    & \leq 1 + m(P_2,\alpha(\omega,G)|_{P_2},d) \\
				    & \leq 1 + m(P_1, \alpha(\omega|_{P_1},P_1),d) \\
				    & = m(P_1,\omega|_{P_1},d).
\end{align*}				    

Now suppose that $\alpha$ does not change the colour of any such pair of vertices on $P_1$.  Then $\alpha(\omega,G)|_{P_1} = \alpha(\omega|_{P_1},P_1)$ and so $m(P_1, \alpha(\omega,G)|_{P_1},d) = m(P_1,\omega|_{P_1},d) - 1$.  Applying the inductive hypothesis in this case then gives
\begin{align*}
m_{G,\omega}(u,v,d) & \leq 1 + m_{G,\alpha(\omega,G)}(u,v,d) \\
					& \leq 1 + m(P_1, \alpha(\omega,G)|_{P_1},d) \\
					& = m(P_1,\omega|_{P_1},d).
\end{align*}

Thus we have
$$m_{G,\omega}(u,v,d) = \min_{P \in \mathcal{P}_G(u,v)} m(P,\omega|_P,d),$$
as required.
\end{proof}

Furthermore, we now see that if $P$ is a critical $u$-$v$ path with respect to $S$, an optimal sequence to link $u$ and $v$, then all moves in $S$ are played on $P$.

\begin{lma}
Let $G$ be any connected graph with colouring $\omega$, $S$ an optimal sequence of moves to link $u,v \in V$ in $G$, and $P$ a critical $u$-$v$ path with respect to $S$.  Then all moves of $S$ are played in areas intersecting $P$.
\label{on-good-path}
\end{lma}

\begin{proof}
First note that, by Lemma \ref{min-path}, $|S| \leq m(P,\omega)$.  Let $S'$ be the subsequence of $S$ consisting of moves played in areas intersecting $P$, and without loss of generality assume that all moves in $S'$ are in fact played on $P$.  Then, as $P$ is critical, $S'$ played on the path $P$ (considered as a separate graph) makes $P$ monochromatic, and so $|S'| \geq m(P,\omega)$.  But then
$$|S| \leq m(P,\omega) \leq |S'| \leq |S|,$$
so we must have equality throughout.  In particular, $|S'|=|S|$ and hence all moves of $S$ are played in areas intersecting $P$.
\end{proof}

Next we show that it cannot be harder to connect a pair of vertices in a larger graph.

\begin{cor}
Suppose $G$ is any connected graph, with colouring $\omega$ from colour-set $C$, and let $H$ be a connected subgraph of $G$, $d \in C$ and $u,v \in V(H)$.  Then 
$$m_{H,\omega}(u,v,d) \geq m_{G,\omega}(u,v,d).$$
\label{general-subgraph}
\end{cor}

\begin{proof}
As $H$ is a subgraph of $G$, it is clear that $\mathcal{P}_H(u,v) \subseteq \mathcal{P}_G(u,v)$.  Thus, by Lemma \ref{min-path}, we have
$$m_{H,\omega}(u,v,d) = \min_{P \in \mathcal{P}_H(u,v)} m_{P,\omega}(P,\omega,d) \geq \min_{P \in \mathcal{P}_G(u,v)} m_{P,\omega}(P,\omega,d) = m_{G,\omega}(u,v,d).$$
\end{proof}

Our final auxiliary result before the main theorem of this section concerns the additivity of connection times.

\begin{cor}
Let $G=(V,E)$ be any connected graph with colouring $\omega$ from colour-set $C$, and let $u,v \in V$, $xy \in E$ and $d \in C$.  Then
$$m(u,v,d) \leq m(u,x,d) + m(y,v,d).$$
\label{triangle}
\end{cor}

\begin{proof}
By Lemma \ref{min-path}, there exist $u$-$x$ and $v$-$y$ paths $P_{ux}$ and $P_{yv}$ in $G$ such that $m_G(u,x,d) = m(P_{ux},\omega|_{P_{ux}},d)$ and $m(y,v,d) = m(P_{yv},\omega|_{P_{yv}},d)$.  Then $P_{ux}P_{yv}$ gives a $u$-$v$ walk in $G$, and so we can obtain a $u$-$v$ path $P_{uv}$ in $G$ by deleting some vertices from $P_{ux}P_{yv}$ (and joining their neighbours).  Then
\begin{align*}
m(u,v,d) & \leq m(P_{uv},\omega|_{P_{uv}},d) &\mbox{by Lemma \ref{min-path}} \\
	     & \leq m(P_{ux}P_{yv},\omega|_{P_{ux}P{yv}},d) &\mbox{by Lemma \ref{monotonicity} } \\
	     & \leq m(P_{ux},\omega|_{P_{ux}}, d) + m(P_{yv},\omega|_{P_{yv}}, d) &\mbox{by Lemma \ref{cat-paths} } \\
	     & = m(u,x,d) + m(y,v,d) &\mbox{by choice of $P_{ux}$,$P_{yv}$.}
\end{align*}
\end{proof}

Using these results, we now consider how to calculate the minimum number of moves required to link all pairs of vertices in an arbitrary graph.  This problem is similar to the \emph{all-pairs shortest path} problem, which can be solved in time $O(|V|^3)$ using the Floyd-Warshall algorithm, as described in \cite{clr90}.  Here, however, the situation is somewhat more complex: firstly, we need to consider the different costs associated with linking pairs of vertices in different colours, and secondly we cannot simply add costs when we concatenate paths.  These factors lead to the greater complexity of our algorithm.

\begin{thm}
Let $G=(V,E)$ be a connected graph with colouring $\omega$ from colour-set $C$.  Then we can compute $m(u,v)$ for every pair $(u,v) \in V^{(2)}$ in time $O(|V|^3|E||C|^2)$.
\label{connection-time}
\end{thm}

\begin{proof}
We begin by observing that, for any $v \in V$ and $d \in C$, we have
\[
m(v,v,d) = \begin{cases}
		0 & \text{if $v$ has colour $d$}, \\
		1 & \text{otherwise}
	      \end{cases}
\]
We then claim that if we define, for all $u,v \in V$ and $d \in C$,
\begin{align}
m^*(u,v,d) = \min_{xx' \in E} \{& m(u,x,d) + m(x',v,d), \nonumber \\
			          & \min_{d' \in C} \{1 + m(u,x,d') + m(x',v,d')\}\},
\label{gct-recursion}			      
\end{align}
then $m(u,v,d) = m^*(u,v,d)$.

First, let us show that $m^*(u,v,d)$ gives an upper bound on $m(u,v,d)$.  By Corollary \ref{triangle}, $m(u,v,d) \leq m(u,x,d) + m(x',v,d)$ for any edge $xx' \in E$.  Note that
\begin{equation}
m(u,v,d) \leq 1 + m(u,v,d'),
\label{change-colour}
\end{equation}
since with one additional move we can change the colour of the monochromatic area containing $u$ and $v$ to $d$.  So Corollary \ref{triangle} further shows that, for any $xx' \in E$ and $d' \in C$,
$$m(u,v,d) \leq 1 + m(u,v,d') \leq 1 + m(u,x,d') + m(x',v,d').$$
So, taking the minimum over all such possibilities, we have $m(u,v,d) \leq m^*(u,v,d)$, as required.

We now proceed to show the reverse inequality.  By Lemma \ref{min-path}, there exists some $u$-$v$ path $P$ in $G$ so that $m(u,v,d) = m_{P,\omega|_P}(P,\omega|_P,d)$.  Suppose $S$ is an optimal sequence to make the isolated path $P$ (with colouring $\omega|_P$) monochromatic with colour $d$ (so that $m(u,v,d) = |S|$), and consider the three possibilities for the last move of $S$.
\begin{enumerate}

\item The last move links two monochromatic segments of $P$.  Without loss of generality, suppose $xx' \in E(P)$ is such that the segment from $u$ to $x$ has colour $d$ and the segment from $x'$ to $v$ has colour $d'$.  Then $xx' \in E(G)$ and we have
\begin{align*}
|S| & \geq 1 + m_{P,\omega|_P}(u,x,d) + m_{P,\omega|_P}(x',v,d') \\
	& \geq m_{P,\omega|_P}(u,x,d) + m_{P,\omega|_P}(x',v,d) & \mbox{by (\ref{change-colour})}\\
	& \geq m_{G,\omega}(u,x,d) + m_{G,\omega}(x',v,d) & \mbox{by Corollary \ref{general-subgraph}} 
\end{align*}
so certainly
$$m(u,v,d) = |S| \geq \min_{xx' \in E} \{m(u,x,d) + m(x',v,d)\}.$$

\item The last move links three monochromatic segments of $P$, of which the end two must already have colour $d$.  Suppose $xx', yy' \in E(P)$ are such that the segment from $u$ to $x$ has colour $d$, that from $x'$ to $y'$ has colour $d'$, and the final segment from $y$ to $v$ has colour $d$.  Then $xx', yy' \in E(G)$ and 
\begin{align*}
|S| & \geq 1 + m_{P,\omega|_P}(u,x,d) + m_{P,\omega|_P}(x',y',d') + m_{P,\omega|_P}(y,v,d) \\
      & \geq m_{P,\omega|_P}(u,x,d) + m_{P,\omega|_P}(x',y',d) + m_{P,\omega|_P}(y,v,d) & \mbox{by (\ref{change-colour})}\\
      & \geq m_{P,\omega|_P}(u,x,d) + m_{P,\omega|_P}(x',v,d) & \mbox{by Corollary \ref{triangle}} \\
      & \geq m_{G,\omega}(u,x,d) + m_{G,\omega}(x',v,d) & \mbox{by Corollary \ref{general-subgraph}}
\end{align*}
so again,
$$m(u,v,d) = |S| \geq \min_{xx' \in E} \{m(u,x,d) + m(x',v,d)\}.$$

\item $P$ is already monochromatic, and the final move changes its colour to $d$.  In this case, $|S| \geq 1 + m_{P,\omega|_P}(u,v,d')$, for some $d' \in C$.  Note that in an optimal sequence to flood $P$ with colour $d$, $P$ cannot be monochromatic before the penultimate move (otherwise we could obtain a shorter sequence by changing the colour to $d$ immediately).  So in this case the second last move must have linked either two or three monochromatic segments of $P$, and so by the two cases above we have
$$m_{P,\omega}(u,v,d') \geq m_{G,\omega}(u,v,d') \geq \min_{xx' \in E} \{m_{G,\omega}(u,x,d') + m_{G,\omega}(x',v,d')\}.$$
Thus
$$|S| \geq 1 + m_{P,\omega}(u,v,d') \geq 1 + \min_{xx' \in E} \{m_{G,\omega}(u,x,d') + m_{G,\omega}(x',v,d')\},$$
and certainly
$$m(u,v,d) = |S| \geq 1 + \min_{\substack{xx' \in E \\ d' \in C}} \{m(u,x,d') + m(x',v,d')\}.$$
\end{enumerate}
So in all cases we have $m(u,v,d) \geq m^*(u,v,d)$, implying that we do indeed have $m^*(u,v,d) = m(u,v,d)$ for all $u,v \in V$ and $d \in C$, as required.  

In our dynamic program, we initialise values of $m(v,v,d)$ as described above, and set all other values to $\infty$.  For any $u,v \in V$ and $d \in C$, let us define $l(u,v,d)$ to be the minimum length of a $u-v$ path $P$ such that $m(u,v,d) = m_P(P,\omega|_P,d)$, and note that initially $m(u,v,d)$ is calculated correctly if $l(u,v,d)=0$.  Further note that, by the reasoning above, we calculate $m(u,v,d)$ correctly if we consider only triples $(x,y,d')$ in $m^*(u,v,d)$ for which $l(x,y,d') < l(u,v,d)$.  Thus we see inductively that, after the $k^{th}$ iteration, $m(u,v,d)$ is calculated correctly whenever $l(u,v,d) \leq k$.  But for any $u,v \in V$ and $d \in C$, we must have $l(u,v,d) \leq |V|$, so certainly $|V|$ iterations will suffice.  

At each iteration we compute $|V|^2 \cdot |C|$ values of $m(u,v,d)$, and for each one we need to consider $|E|$ possible edges and $|C|$ possible colours, so each iteration takes time $O(|V|^2|E||C|^2)$.  Thus, as we need a total of $|V|$ iterations, the entire computation takes time $O(|V|^3|E||C|^2)$.
\end{proof}

We obtain an easy corollary by applying this result to the special case in which the graph in question is a path.

\begin{cor}
FREE-FLOOD-IT can be solved for any path $P$ in time $O(|P|^6)$, and $c$-FREE-FLOOD-IT can be solved in time $O(|P|^4)$.
\label{polypath}
\end{cor}

\begin{proof}
Let $P$ be a path with colouring $\omega$, and let $u$ and $v$ be the two end-vertices of $P$.  Then a sequence of moves $S$ makes $P$ monochromatic if and only if it links $u$ and $v$, so $m(P,\omega) = m_{P,\omega}(u,v)$.  But by Theorem \ref{connection-time}, we can compute $m_{G,\omega}(u,v)$ for two vertices in any arbitrary graph in time $O(|V|^3|E||C|^2)$.  As $P$ has $O(|P|)$ edges, and we cannot possibly have a colour-set of size greater than $|P|$, the complexity of this algorithm is bounded by $O(|P|^6)$, or $O(|P|^4)$ if the colour-set has fixed size.
\end{proof}

We can also apply this result to the free variant of the game played on rectangular boards of fixed height.  It follows immediately from Corollary \ref{polypath} that FREE-FLOOD-IT restricted to $1 \times n$ boards can be solved in time $O(n^6)$ (and $c$-FREE-FLOOD-IT in time $O(n^4)$).  A further corollary is an additive approximation to $c$-FREE-FLOOD-IT played rectangular boards of fixed height $k$.

\begin{cor}
For any fixed $k$, we can compute a constant additive approximation to $c$-FREE-FLOOD-IT, restricted to $k \times n$ boards, in time $O(n^4)$.
\label{kxn-approx}
\end{cor}
\begin{proof}
Let $B$ be a $k \times n$ Flood-It board, with at most $c$ colours, and let $u$ (respectively $v$) be a vertex in $G_B$ corresponding to a square incident with the left-hand (respectively right-hand) edge of the board.  Suppose the sequence of moves $S$ floods $B$. Then $S$ clearly links $u$ and $v$, so we have $|S| \geq m_{G_B,\omega_B}(u,v)$.  But observe also that one strategy to flood the board would be to create a monochromatic path from $u$ to $v$, and then cycle through all $c$ colours at most $k-1$ times to absorb all remaining regions.  Thus we have $m(B) \leq m_{G_B,\omega_B}(u,v) + c(k-1)$.  Hence
\begin{equation*}
m_{G_B,\omega_B}(u,v) \leq m(B) \leq m_{G_B,\omega_B}(u,v) + c(k-1),
\label{path-approx}
\end{equation*}
and $m_{G_B,\omega_B}(u,v)$ gives an additive approximation to $m(B)$.

But by Theorem \ref{connection-time}, we can compute $m_{G_B,\omega_B}(u,v)$ in time $O(n^4)$ for fixed $k$ and $c$, thus obtaining in polynomial time an additive approximation to $m(B)$.
\end{proof}

\section{Rectangular boards of constant height}
\label{complexity}

In contrast to our approximation result in Corollary \ref{kxn-approx}, we see in this section that, even for small values of $k$, it remains NP-hard to solve flood-filling problems \emph{exactly} on $k \times n$ boards.

In particular, we show that both 4-FIXED-FLOOD-IT and 4-FREE-FLOOD-IT remain NP-hard when restricted to $3 \times n$ boards.  This improves on the result of Clifford, Jalsenius, Montanaro and Sach in \cite{clifford} that FREE-FLOOD-IT remains NP-hard on such boards.  Both our results are proved by means of reductions from the decision version of Shortest Common Supersequence (SCS), shown to be NP-complete over a binary alphabet by R\"{a}ih\"{a} and Ukkonen in \cite{raiha81}.

Suppose we have an SCS instance consisting of $k$ strings $s_1, \ldots, s_k$ over a binary alphabet $\Sigma =\{1,2\}$, where each string has length at most $w$, and an integer $l$.  The problem is to determine whether $s_1, \ldots, s_k$ have a common supersequence of length at most $l$.  We will construct $3 \times n$ boards $B$ and $B'$ for the 4-FIXED-FLOOD-IT and 4-FREE-FLOOD-IT problems respectively (each using colours $\{1,2,3,4\}$), so that $m(B), m(B') \leq 2l + 3$ if and only if $s_1, \ldots s_k$ have a common supersequence of length at most $l$.

\subsection{The 4-FIXED-FLOOD-IT case}

We prove the following theorem.
\begin{thm}
4-FIXED-FLOOD-IT remains NP-hard when restricted to $3 \times n$ boards.
\label{3xn-fixed}
\end{thm}
To show the reduction, we construct a $3 \times n$ Flood-It board with four colours as follows.  For each $s_i$, we include a $2 \times (2|s_i| + 1)$ gadget $G_i$ as illustrated in Figure \ref{G_i}, where $s_i[j]$ denotes the $j^{th}$ character of $s_i$.

\begin{figure} [h]
\centering
\includegraphics[width=0.8\linewidth]{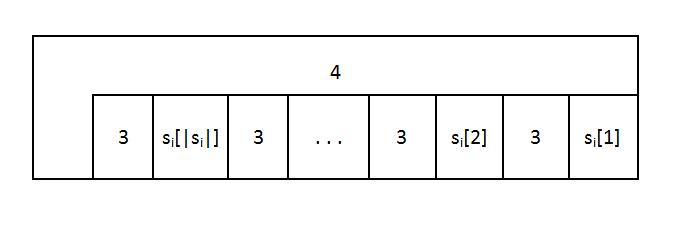}
\caption{The gadget $G_i$}
\label{G_i}
\end{figure}

We then place these in a $3 \times n$ board filled with colour 3 as illustrated in Figure \ref{board-B}, and add a section $R$ at the end, where $R$ is as shown in Figure \ref{rectangle-R}.  Note we can take $n \leq k(2w+2) + 2l +3$.

\begin{figure} [h]
\centering
\includegraphics[width=0.7\linewidth]{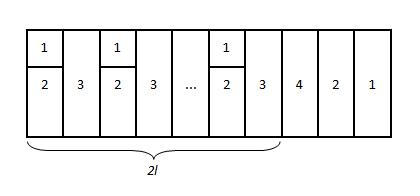}
\caption{The rectangle $R$}
\label{rectangle-R}
\end{figure}

\begin{figure} [h]
\centering
\includegraphics[width=0.8\linewidth]{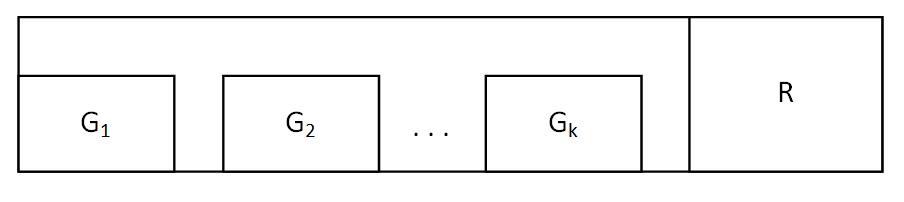}
\caption{The board $B$}
\label{board-B}
\end{figure}

We now show, in the next two lemmas, that $s_1,\ldots,s_k$ have a common supersequence of length at most $l$ if and only if we can flood $B$ in $2l + 3$ steps.

\begin{lma}
If $s_1, \ldots, s_k$ have a common supersequence of length at most $l$, then we can flood the board $B$ (starting from the top left) in $2l + 3$ steps.
\label{cs=>strat}
\end{lma}

\begin{proof}
Let $a_1 \ldots a_l$ be a common supersequence of length exactly $l$ (padding a shortest sequence with $1$s if necessary).  Then we claim that the sequence of moves $a_13a_23a_33\ldots a_{l-1}3a_l3421$, of length $2l + 3$, floods the board.  First observe that this sequence floods $R$: each move extends the external area into $R$ by at least part of one column, and the final two moves of colours $2$ and $1$ respectively will flood any remaining unflooded partial columns.  But this sequence will also flood $G_i$ for each $i$: $s_i$ is a subsequence of $a_1 \ldots a_l$ so $s_i[1]3s_i[2]\ldots s_i[|s_i|]3$ is a subsequence of $a_13a_23\ldots a_l3$, and the first $2l$ moves will flood all of $G_i$ not coloured 4, before the $(2l+1)^{st}$ move floods the region coloured 4.  So this sequence of $2l + 3$ moves does indeed flood $B$.
\end{proof}

\begin{lma}
If we can flood $B$ in at most $2l + 3$ steps, starting from the top left corner at each move, then $s_1, \ldots s_k$ have a common supersequence of length at most $l$.
\label{strat=>cs}
\end{lma}

\begin{proof}
First observe that we cannot flood $R$ from the outside in fewer than $2l + 3$ steps, as each move can only move the boundary of the external area to the right by one column.  Moreover, any sequence of $2l + 3$ moves that floods $R$ must consist of $l$ $1$s or $2$s, alternated with $3$s, then finally $4,2,1$.

Suppose such a sequence $c_1, \ldots c_{2l+3}$ also floods every $G_i$.  Note that the external area never has colour $3$ after the only move of colour $4$.  So, in order to flood the leftmost square of colour $3$ in each $G_i$, we must in fact flood the bottom row of each $G_i$ sequentially from the right, and moreover we must have flooded this row by the end of the $(2l)^{th}$ move.  But then, for each $i$, $s_i$ must be a subsequence of $c_1, \ldots c_{2l}$ restricted to $\{1,2\}$, which is a sequence of length $l$.  Hence we have a common supersequence of $s_1, \ldots s_k$ of length $l$, as required.
\end{proof}

\begin{proof}[Proof of Theorem \ref{3xn-fixed}]
The reduction from Shortest Common Supersequence follows immediately from Lemmas \ref{cs=>strat} and \ref{strat=>cs}.
\end{proof}

\subsection{The 4-FREE-FLOOD-IT case}

We now prove an analogous theorem for the free variant of the game.

\begin{thm}
4-FREE-FLOOD-IT remains NP-hard when restricted to $3 \times n$ boards.
\label{3xn-free}
\end{thm}

The construction of the Flood-It board $B'$ used to prove Theorem \ref{3xn-free} is very similar to that in the previous section.  The only difference is that we also include a second rectangular section, $R'$, located at the left-hand end of the board (as illustrated in Figure \ref{board-B'}).  $R'$ is identical to $R$ except that it is reflected in a vertical axis.  In this case we can take $n \leq k(2w+2) + 4l + 7$.

\begin{figure} [h]
\centering
\includegraphics[width=0.8\linewidth]{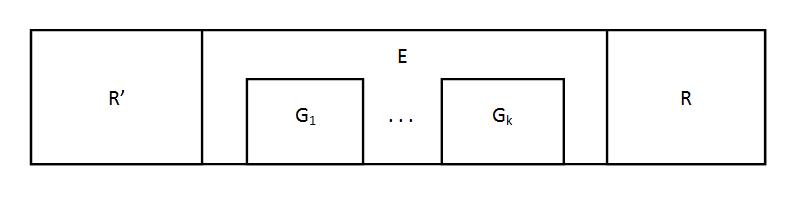}
\caption{The board $B'$}
\label{board-B'}
\end{figure}

It remains to show that $s_1,\ldots,s_k$ have a common supersequence of length at most $l$ if and only if we can flood $B'$ in $2l + 3$ steps.

\begin{lma}
If $s_1, \ldots, s_k$ have a common supersequence of length at most $l$, then we can flood the board $B'$ in $2l+3$ steps.
\label{cs=>freestrat}
\end{lma}

\begin{proof}
We can use exactly the same strategy as in Lemma \ref{cs=>strat}, playing in the external area $E$ at each move.
\end{proof}

\begin{lma}
If we can flood $B'$ in at most $2l+3$ moves, making moves anywhere on the board, then $s_1, \ldots, s_k$ have a common supersequence of length at most $l$.
\label{freestrat=>cs}
\end{lma}

\begin{proof}

First observe that we need at least $2l +3$ moves to flood the board: initially the minimum number of monochromatic areas on any end-to-end path is $4l+7$, and as each move can decrease this by at most two, we do indeed require a minimum of $2l+3$ moves.  Moreover, to achieve this lower bound, every move must reduce the number of monochromatic areas lying on an end-to-end path by exactly two.  One consequence of this is that no move can be played inside any $G_i$.

Another consequence is that we can only make a move of colour 4 once in any optimal sequence: only two regions of this colour lie on any shortest end-to-end path, and so we can only make one colour 4 move that will decrease the path length as required.  However, to flood all the $G_i$, this single move of colour 4 must be played in $E$, so we cannot play colour 4 until the area $A$, containing $E$ and adjacent to both colour 4 regions in $R \cup R'$, has been linked.  It requires at least $2l$ moves to link $A$, so we make at least $2l + 1$ moves up to and including the move of colour 4.  These moves have no effect on the regions of colour 1 and 2 at the ends of the board, and it requires at least two moves to flood these remaining end-regions, so we can only possibly flood $B$ in $2l+3$ moves if we link $A$ (except for regions of colour 4) in exactly $2l$ moves, and then play colours 4,2, and 1 in the external area.

As colour 3 is then never played after colour 4 we see, as before, that the left-most square of colour 3 in each $G_i$ can only be flooded if the bottom row of each $G_i$ is flooded sequentially from the right, and this must be done within the first $2l$ moves.  Thus, if $s$ is the subsequence of the first $2l$ moves consisting of those that are made in an area containing $E$ and are of colour $1$ or $2$, then $s$ is a common supersequence of $s_1, \ldots, s_k$.

To complete the proof it therefore suffices to show that $|s| \leq l$.  But every move in $s$ reduces by two the number of monochromatic areas lying on the shortest end-to-end path, by means of flooding two areas of colour $c \in \{1,2\}$.  Initially there were only $2l + 4$ regions of colour 1 or 2 on any shortest end-to-end path, and four of these we know are not flooded within the first $2l$ moves, so $|s|$ can be at most $l$, as required.
\end{proof}

\begin{proof}[Proof of Theorem \ref{3xn-free}]
The reduction from Shortest Common Supersequence follows immediately from Lemmas \ref{cs=>freestrat} and \ref{freestrat=>cs}.
\end{proof}

\section{Conclusions and open problems}

In the case of the game played on rectangular $k \times n$ boards, we have shown that we can solve FREE-FLOOD-IT, restricted to $1 \times n$ boards, in polynomial time, and also that we can calculate in polynomial time an additive approximation in this case for any fixed $k$.  However, we have demonstrated that $c$-FREE-FLOOD-IT remains NP-hard when restricted to $k \times n$ boards for any $k \geq 3$ and $c \geq 4$.

In the general graph context, we have shown that the connection time between any pair of vertices can be computed in polynomial time.  A natural extension would be to consider the complexity of computing the number of moves required to connect a set of $k$ vertices.
\begin{prob}
Given a graph $G$ and a subset $U \subset V(G)$ of (fixed) size $k$, what is the complexity of determining the minimum number of moves required to create a monochromatic component containing all $u \in U$?
\end{prob}

Very few results are known about which classes of graphs allow a polynomial time algorithm to solve FIXED-FLOOD-IT, c-FIXED-FLOOD-IT, or the
free variants of these problems. However, we make one conjecture.
\begin{conj}
$c$-FREE-FLOOD-IT is polynomially solvable on subdivisions of any fixed graph $H$.
\label{subdivisions}
\end{conj}
Note that Conjecture \ref{subdivisions} would imply that $c$-FREE-FLOOD-IT is solvable in polynomial time on cycles and on trees with only a bounded number of vertices of degree at least three.  The conjecture may in fact hold even if we allow a colour-set of unbounded size.

\end{document}